\documentclass{article}
\usepackage{spconf}
\usepackage{exscale}
\usepackage{amsmath,esint}
\usepackage{amssymb}
\usepackage{graphicx,float}
\usepackage{fixltx2e}
\usepackage{subcaption}
\usepackage{epstopdf}
\usepackage{multirow}
\usepackage{array}
\usepackage{cite}
\usepackage{color}

\newcommand{\bx}{\mathbf{x}}    % transmit signal vector
\newcommand{\bw}{\mathbf{w}}    % weight vector
\newcommand{\bH}{\mathbf{H}}    % channel matrix
\newcommand{\be}{\mathbf{e}}    % eigen vector
\newcommand{\bh}{\mathbf{h}}    % channel vector
\newcommand{\bN}{\mathbf{N}}    % noise covariance
\newcommand{\bC}{\mathbf{C}}    % coupling matrix
\newcommand{\bI}{\mathbf{I}}    % Current vector or identity matrix
\newcommand{\bV}{\mathbf{V}}    % voltage vector
\newcommand{\bS}{\mathbf{S}}    % covariance
\newcommand{\bX}{\mathbf{X}}    % reactance matrix
\newcommand{\bY}{\mathbf{Y}}    % reactance matrix
\newcommand{\bU}{\mathbf{U}}    % auxiliary matrix
\newcommand{\bA}{\mathbf{A}}    % auxiliary matrix
\newcommand{\bB}{\mathbf{B}}    % auxiliary matrix
\newcommand{\define}{\triangleq}    % Bold vector b

\newcommand{\mbbC}{\mathbb{C}}	% field of complex number
\newcommand{\btt}{\bsym{\theta}}

\newcommand{\ben}{\begin{enumerate}} 	  	% 	Begin Enumerate
	\newcommand{\een}{\end{enumerate}} 			% 	End Enumerate
\newcommand{\beq}{\begin{equation}} 	  	% 	Begin Equation
\newcommand{\eeq}{\end{equation}} 			% 	End Equation
\newcommand{\bes}{\begin{equation*}}
\newcommand{\ees}{\end{equation*}}

\newcommand{\bea}{\begin{eqnarray}}		% 	Begin Equation Array
\newcommand{\eea}{\end{eqnarray}} 		% 	End Equation Array
\newcommand{\beas}{\begin{eqnarray*}}
	\newcommand{\eeas}{\end{eqnarray*}}
\newcommand{\ba}{\begin{array}}
	\newcommand{\ea}{\end{array}}
\newcommand{\sbea}{\nopagebreak[3]\samepage\begin{eqnarray}}
\newcommand{\seea}{\end{eqnarray}\pagebreak[0]}
\newcommand{\sbeas}{\nopagebreak[3]\samepage\begin{eqnarray*}}
	\newcommand{\seeas}{\end{eqnarray*}\pagebreak[0]}

\newcommand{\lb}{\label}
\newcommand{\er}[1]{{\rm(\ref{#1})}}

\newcommand{\bit}{\begin{itemize}}
	\newcommand{\eit}{\end{itemize}}
\newcommand{\bsym}{\boldsymbol}

\newcommand{\nn}{\nonumber}

\DeclareMathOperator{\Trace}{Tr}

\newtheorem{theorem}{Theorem}%[section]

\newenvironment{proof}[1][Proof]{\begin{trivlist}
		\item[\hskip \labelsep {\bfseries #1}]}{\end{trivlist}}

\newcommand{\qed}{\nobreak \ifvmode \relax \else
	\ifdim\lastskip<1.5em \hskip-\lastskip
	\hskip1.5em plus0em minus0.5em \fi \nobreak
	\vrule height0.75em width0.5em depth0.25em\fi}

\def\iflatex{\iftrue}
\def\ifcomments{\iffalse}

\begin{document} 
    \title{Antenna Impedance Estimation in Correlated Rayleigh Fading Channels}
    \twoauthors
    {Shaohan Wu}
    {}
    %\\
    %  	Department of Communication Systems Design\\
    %  	Address A-B}
    {Brian L. Hughes}{}
%    {North Carolina State University}
    %  	{North Carolina State University\\
    %  	Department of Electrical and Computer engineering\\
    %  	Address C-D}
    
    \maketitle
    \begin{abstract}
        We formulate antenna impedance estimation in a classical estimation framework under correlated Raleigh fading channels. Based on training sequences of multiple packets, we derive the  ML estimators for antenna impedance and channel variance, treating the fading path gains as nuisance parameters.  These ML estimators can be found via scalar optimization. 
        We explore the efficiency of these estimators against Cramer-Rao lower bounds by numerical examples. The impact of channel correlation on impedance estimation accuracy is investigated.
    \end{abstract}
    
    
    \section{Introduction}
    Antenna impedance matching to the receiver front-end has been shown to significantly impact the capacity and diversity of wireless channels \cite{domi2}.
    This matching becomes challenging as antenna impedance changes with time-varying near-field loading, e.g., human users. 
    To mitigate this change, antenna impedance estimation techniques at mobile receivers have been proposed \cite{hass,wu_pca, wu_ciss18,wu_ciss21,wu_hybrid18,wu_pca_mimo}.
    
    Hassan and Wittneben proposed least-square estimators to jointly estimate the spatial channel and coupling impedance matrices \cite{hass}. Wu and Hughes first derived joint channel and antenna impedance estimators at single-antenna receivers using a hybrid estimation framework\cite{wu_hybrid18}. Wu extended it to multi-antenna receivers \cite{wu_ciss21}. Under classical estimation, the maximum-likelihood (ML) estimators of antenna impedance have been derived under i.i.d. Rayleigh fading, treating channel path gains as nuisance parameters \cite{wu_pca,wu_pca_mimo}. 
    However, the optimal impedance estimator remains unknown when the channel is correlated. In this paper, we fill this gap. 
    
    We formulate antenna impedance estimation in a classical estimation framework under correlated Raleigh fading channels. Based on training sequences of multiple packets, we derive the  ML estimators for antenna impedance and channel variance, treating the fading path gains as nuisance parameters.  These ML estimators can be found via scalar optimization. 
    We explore the performance, e.g., efficiency against Cramer-Rao lower bounds, of these estimators through numerical examples. The impact of channel correlation on impedance estimation accuracy is also investigated.
    
    The rest of the paper is organized as follows. We present the system model in Sec.~\ref{3secII} and derive the ML estimators in Sec.~\ref{sec_ml}. We explore the performance of these estimators through numerical examples in Sec.~\ref{3secV} and conclude in Sec.~\ref{3secVI}.
    
    \section{System Model}
    \lb{3secII}
    Consider a narrow-band, multiple-input, single-output (MISO) channel with $N$ transmit antennas and one receive antenna. 
    Suppose the transmitter sends $L$ packets each with an identical training sequence to the receiver. During transmission, the receiver front-end shifts halfway in the training sequence\cite[eq.~7]{wu_pca}, to observe the unknown antenna impedance. We assume the channel
    is constant {\em within} a packet, but generally varies from packet to packet randomly.
    Under these assumptions, the signal observed during the $k$-th packet can be described by \cite[eq.~10]{wu_pca}, with $K$ assumed even, 
    \beq\label{3observationsMP}
    u_{k,t}  \ = \ \begin{cases}
        \bh_k^T \bx_t + n_{k,t} \ , & 1 \leq t \leq K/2 \ , \\
        F \bh_k^T \bx_t  + n_{k,t} \ , & K/2 < t \leq K \ ,
    \end{cases}
    \eeq
    where $F$ is a function of the unknown antenna impedance \cite[eq.~11]{wu_pca}, $\bh_k$ is the channel during $k$-th packet, and the noise $n_{k,t} \sim\mathcal{CN}(0, 1)$ is i.i.d..
    We can express \er{3observationsMP} in matrix form, 
    \beq\label{3observationsMP2}
    \bU_1 = \bH \bX_1 + \bN_1 \ , ~~~\bU_2 = F \bH \bX_2 + \bN_2
    \eeq
    where $\bX_1 \define [ \bx_1 , \ldots, \bx_{K/2}]$, $\bX_2 \define [ \bx_{K/2+1} , \ldots, \bx_K]$,
    \beq\label{H_ch3}
    \bH \ \define \ 
    [ \bh_1 , \ldots, \bh_L]^T \in \mbbC^{L\times N} \ . 
    \eeq 
    It follows $\bN_1$ and $\bN_2$ are independent random matrices with i.i.d. $\mathcal{CN}(0, 1)$ entries. Note the horizontal dimension of $\bH$ represents space, while the vertical dimension is time.  Here $\bH$ models Rayleigh fading path gains which are uncorrelated in space but generally correlated in time. This implies the columns of $\bH$ are i.i.d. zero-mean, complex Gaussian random vectors with a temporal correlation matrix $\sigma_h^2 \bC_\bH$. We assume the correlation matrix $\bC_\bH$ is known but the power $\sigma_h^2$ is unknown.  As in our prequel, we assume the known sequences $\bX_1$ and $\bX_2$ are equal-energy and orthogonal over the first and last $K$ symbols \cite[eq.~16]{wu_pca},
    \beq\label{XsMP}
    \bX_1\bX_1^H \ = \ \bX_2\bX_2^H \ = \ \left( \frac{PK}{2N} \right) \bI_N \ .
    \eeq
    
    \section{Maximum-Likelihood Estimators}\label{sec_ml}
    The goal of this paper is to derive optimal estimators for 
    \beq\label{eq_theta}
    \btt \ \define \begin{bmatrix}
        F & \sigma_h^2
    \end{bmatrix}^T
    \eeq
    based on the observations \er{3observationsMP2}. To this end, we leverage the classical estimation framework by treating $\bH$ as a nuisance parameter. We first prove the sufficiency of observations. 
    
    \begin{theorem}[Sufficient Statistics]\lb{sslemmaMP}  Given $\bU_1$ and $\bU_2$ defined in
        \eqref{3observationsMP2}, where $\bX_1,\bX_2$ are known training sequences \er{XsMP} and $\btt$ in \er{eq_theta} are unknown, then
        \begin{eqnarray}\label{ysMP}
            \bY_1 \ \define \ \sigma^2 \bU_1 \bX_1^H \ , \ \ \bY_2 \ \define \ \sigma^2\bU_1 \bX_2^H  ,
        \end{eqnarray}
        are sufficient statistics to estimate $\btt$, where we define
        \beq\label{eq_sigma}
        \sigma^2 \ \define \ \frac{2N}{PK} \ .
        \eeq
        Moreover, $\bY_1-\bH$ and $\bY_2-F\bH$ are independent random matrices with i.i.d. $\mathcal{CN}(0, \sigma^2)$ entries. \hfill $\diamond$
    \end{theorem}
    
    \begin{figure*}
        \vspace{-16pt}
        \bea\label{3LemmaMP}
        &&\pi^{LK} \cdot p\left( \bU_1 , \bU_2 ;\btt \right) \ = \ E_\bH\left[ \exp \left( - \left\lVert \bU_1 - \bH \bX_1 \right\rVert^2
        -\left\lVert \bU_2 - F \bH \bX_2 \right\rVert^2 \right) \right]  \nn\\
        &=&  E_\bH \left[ \exp \Biggl( \frac{1}{\sigma^2}\biggl\{ 2 {\rm ReTr}[ \bH^H \bY_1 ] + 2 {\rm ReTr}[ F^* \bH^H \bY_2 ] - (1 + |F|^2) \left\lVert \bH \right\rVert^2 \biggr\} \Biggr) \right] \exp\left(-  \left\lVert \bU_1 \right\rVert^2-\left\lVert \bU_2 \right\rVert^2 \right)  \ , 
        \eea
        \vspace{-16pt}
    \end{figure*}
    
    \begin{proof}
        From \er{3observationsMP2} and \er{XsMP}, we have $\bY_1 = \bH + \sigma^2\bN_1 \bX_1^H$. Note the rows
        of $\sigma^2\bN_1 \bX_1^H$ are i.i.d. with covariance $\sigma^4 \bX_1 \bX_1^H = \sigma^2 \bI_N$, due to \er{XsMP}. Similarly, $\bY_2 = F\bH + \sigma^2\bN_2 \bX_2^H$, where the last matrix has i.i.d $\mathcal{CN}(0, \sigma^2)$ entries. So $\bY_1-\bH$ and $\bY_2-F\bH$ are independent random matrices with i.i.d. $\mathcal{CN}(0, \sigma^2)$ entries.
        
        From the Neyman-Fisher Theorem \cite[pg.~117]{kay}, to prove sufficiency it suffices to show $p( \bU_1 , \bU_2 ; \btt)$ can be factored into a product $g( \bY_1 , \bY_2, \btt )f(\bU_1, \bU_2 )$, where $f$ does not depend on $\bY_1 , \bY_2$ or $\btt$, and $g$ does not depend on $\bU_1, \bU_2$. We can express this pdf in terms of the conditional pdf as
        \bea
        p( \bU_1 , \bU_2 ; \btt) \ = \ E_\bH \biggl[ p( \bU_1 , \bU_2 | \bH ; \btt ) \biggr] \ ,\nn
        \eea
        where $E_\bH [ \cdot ]$ denotes expectation with respect to $\bH$.
        Since $\bU_1 , \bU_2$ are conditionally independent given $\bH$, we can simplify the pdf in to \er{3LemmaMP}, 
        where $\left\lVert \bA \right\rVert^2 = {\rm Tr}[ \bA^H \bA ]$ denotes the Frobenius norm.
        Note identities $2{\rm ReTr}[\bA]={\rm Tr}[\bA]+{\rm Tr}[\bA^H]$ and ${\rm Tr}[\bA \bB] = {\rm Tr}[\bB \bA]$ are used, along with \er{XsMP} and \er{eq_sigma}. 
        In \er{3LemmaMP}, denote the first factor by $\pi^{LK} g( \bY_1 , \bY_2, \btt)$, and the second by $f(\bU_1, \bU_2 )$. 
        Note $f$ does not depend on $\bY_1 , \bY_2$ or $\btt$, while $g$ depends on $\bY_1 , \bY_2, F$ and $\sigma_h^2$ (through the expectation), but not $\bU_1, \bU_2$. \hfill $\diamond$
        %This completes the proof. 
    \end{proof}
    
    We now present maximum-likelihood (ML) estimators for $\btt$ defined in \er{eq_theta}
    using sufficient statistics \er{ysMP}. By definition, the ML estimators maximize the likelihood function, i.e., 
    \bea\label{3MLdefMP}
    \hat{\btt}_{ML} \ \define \ {\rm arg} \max_{\btt} p( \bY_1 , \bY_2 ; \btt) \ .
    \eea
    Based on this criterion, we show in the following theorem the ML estimators can be calculated directly after a scalar optimization.
    
    \begin{theorem}[Multiple-Packet ML Estimators]\label{3ThrmMLMP}
        Let $\bY_1$ and $\bY_2$ be the sufficient statistics in \er{ysMP}, where $\btt$ in \er{eq_theta} are unknown constants. Consider the matrix
        \beq\label{samp_covMP}
        \bS ( \mu ) \ \define \ \frac{1}{N}\begin{bmatrix}
            S_{11} ( \mu) & S_{12} ( \mu) \\
            S_{21} ( \mu) & S_{22} ( \mu)
        \end{bmatrix}\ .
        \eeq
        where we define for $1 \leq i,j \leq 2$, 
        \beq
        S_{ij} ( \mu ) \ \define \ {\rm Tr}  \left[ \mu \bC_\bH \left(\mu \bC_\bH + \sigma^2 \bI_L \right)^{-1} \bY_i \bY_j^H \right].
        \eeq
        With $\sigma^2$ in \er{eq_sigma} and $\bC_\bH$ known, we define a scalar optimization problem
        \bea\label{3muhatMP}
        \hat{\mu} \ \define
        {\rm arg} \ \max_{\mu \geq 0} \  \ \left[ \eta ( \mu ) - \sigma^2 \ln {\rm det} [\mu \bC_\bH + \sigma^2 \bI_L ] \right] \ ,
        \eea
        where $\eta ( \mu )$ is the largest eigenvalue of $\bS ( \mu)$ in \er{samp_covMP}:
        \beq\label{3etaMP}
        \eta ( \mu ) \ \define \ \frac{S_{22}+S_{11} + \sqrt{(S_{11}-S_{22} )^2+4|S_{12}|^2}}{2} \ .
        \eeq
        Let $\hat{\be_1}=[E_1,E_2]^T$ be any unit eigenvector of $\bS ( \hat{\mu})$ corresponding to the eigenvalue $\eta ( \hat{\mu} )$.
        Then the maximum-likelihood estimates of $F$ and $\sigma_h^2$ are given by
        \beq\label{3MLEMP}
        \hat{\btt}_{ML} \ = \ \begin{bmatrix}
            \hat{F}_{ML}\\
            \hat{\sigma_h^2}
        \end{bmatrix}
        \ = \ \begin{bmatrix}
            {E_2}/{E_1}\\
            |E_1|^2 \hat{\mu}
        \end{bmatrix} \ ,
        \eeq
        provided
        $E_1\neq 0$. For $E_1 =0$ and $\hat{\mu} > 0$ the likelihood is maximized in the limit as $F \rightarrow \infty$.  $\hfill\diamond$
    \end{theorem}
    
    \begin{proof}  
        For any matrix $\bA$, denote the $kj$-th element and $k$-th row by $[\bA]_{kj}$ and $[\bA]_{k}$, respectively. Let $\bC_\bH = \bV^H {\rm diag}[ \lambda_1 , \ldots , \lambda_L] \bV$ be an eigen-decomposition of $\bC_\bH$, where $\lambda_1 \geq \ldots \geq \lambda_L \geq 0$ are eigenvalues of $\bC_\bH$, and $\bV$ is a unitary matrix such that $\bV \bV^H = \bV^H \bV = \bI_L$.
        It follows the elements of $\bV \bH$  are independent with $[\bV \bH]_{kj} \sim\mathcal{CN}(0, \sigma_h^2 \lambda_k)$.
        
        For $1 \leq k \leq L$, let $\bw_{k1} \define [ \bV \bY_1]_k$ and $\bw_{k2} \define [ \bV \bY_2]_k$. From Theorem~\ref{sslemmaMP}, $\bw_{k1}$ and $\bw_{k2}$ are conditionally independent given $[\bV\bH]_k$, with conditional distributions
        $\bw_{k1} \sim\mathcal{CN}([\bV\bH]_k, \sigma^2 \bI_N)$
        and $\bw_{k2} \sim\mathcal{CN}(F[\bV\bH]_k, \sigma^2 \bI_L)$. Since  $[\bV\bH]_k \sim\mathcal{CN}({\bf 0}_N , \sigma_h^2 \lambda_k \bI_N )$ is independent of the noise in \er{ysMP}, it follows $\bw_k \define ( \bw_{k1}, \bw_{k2})^T$ is a zero-mean Gaussian random vector with covariance
        \bea
        \bC_{{\bw}_k} \ \define \ E \left[ \bw_k^H \bw_k \right] &= &    \bC_k \otimes \bI_N  \ ,
        \eea
        where we define
        \bea\label{3Ckdef}
        \bC_k \ \define \  \begin{bmatrix}
            \sigma_h^2 \lambda_k + \sigma^2 & \sigma_h^2 F^* \lambda_k \\
            \sigma_h^2 F \lambda_k & \sigma_h^2 |F|^2 \lambda_k + \sigma^2
        \end{bmatrix}  \ .
        \eea
        Note 
        $\bC_k$ can be written in terms of its eigensystem as
        \bea
        \bC_k \ = \ \mu_{k1}\be_1\be_1^H +  \mu_2\be_2\be_2^H \ ,
        \eea
        where $\mu_{k1} \geq \mu_2$ are the ordered eigenvalues and $\be_1, \be_2$ are the associated unit eigenvectors. From \er{3Ckdef}, it is easy to verify the following explicit formulas,
        \bea\lb{eigensystemMP}
        \mu_{k1} & = & \mu \lambda_k + \sigma^2 \ , ~~~\mu_2 \ = \ \sigma^2 \\
        \be_1 & = & \frac{1}{\sqrt{1+|F|^2}}\begin{bmatrix}
            1\\
            F
        \end{bmatrix} \ , ~ \be_2 =  \frac{1}{\sqrt{1+|F|^2}}\begin{bmatrix}
            -F^*\\
            1
        \end{bmatrix} \ . \nn
        \eea
        where $\mu \define \sigma_h^2 (1+|F|^2)$. Note only $\mu_{k1}$ depends on $k$. As in \cite[eq.~31]{wu_pca}, we can simplify $\ln p\left(\bw_{k1} , \bw_{k2} ;\btt\right)$ into, 
        \bea\label{llfMP}
        %    && \\ %& = &
        %-N\ln \det(\pi \bC_k) -N\Trace\left[\bS_k\bC_k^{-1}\right] \nn\\
        %	&=&
        %-N\ln(\pi \mu_{k1}\mu_2)  -\frac{N\be_1^H\bS\be_1}{\mu_1}-\frac{N\be_2^H\bS\be_2}{\mu_2}\nn\\
        %	&=& -N\ln(\pi \mu_{k1}\mu_2) +  \left(\mu_2^{-1}-\mu_{k1}^{-1}\right)N\be_1^H\bS_k\be_1
        %-{\mu_2^{-1}}{N\Trace[\bS_k]} \nn \\
        %&=& -N\ln(\pi \sigma^2 [ \mu \lambda_k + \sigma^2])
        %+  \frac{N \mu \lambda_k}{\sigma^2 [ \mu \lambda_k + \sigma^2]} \be_1^H\bS_k\be_1-\sigma^{-2}{N\Trace[\bS_k]}\ \nn \\
        B_k + \frac{N}{\sigma^2} \left[ \frac{\mu \lambda_k}{\mu \lambda_k + \sigma^2} \be_1^H\bS_k\be_1 - \sigma^2 \ln(\mu \lambda_k + \sigma^2) \right]\nn \ ,
        \eea
        where $B_k$ does not depend on $\mu$ or $\be_1$ and we define
        \bea\label{samp_covMPk}
        &&\bS_k \ \define \ \frac{1}{N}\begin{bmatrix}
            \bw_{k1}\bw_{k1}^H& \bw_{k1}\bw_{k2}^H\\
            \bw_{k2}\bw_{k1}^H & \bw_{k2}\bw_{k2}^H
        \end{bmatrix}\nn\\
        &= & \frac{1}{N}\begin{bmatrix}
            [ \bV \bY_1 \bY_1^H \bV^H]_{kk} & [ \bV \bY_1 \bY_2^H \bV^H]_{kk}\\
            [ \bV \bY_2 \bY_1^H \bV^H]_{kk} & [ \bV \bY_2 \bY_2^H \bV^H]_{kk}
        \end{bmatrix} \ .
        \eea
        Since $\bw_1, \ldots , \bw_L$ are independent, the
        joint probability of $\bY_1$ and $\bY_2$ is then given by
        \bea\label{3llrMP}
        &&\ln p( \bY_1 , \bY_2 ; \btt) \ = \ \sum_{k=1}^L \ln p\left(\bw_{k1} , \bw_{k2} ;\btt\right)  \\
        %-N\ln \det(\pi \bC_k) -N\Trace\left[\bS_k\bC_k^{-1}\right] \nn\\
        %	&=&
        %-N\ln(\pi \mu_{k1}\mu_2)  -\frac{N\be_1^H\bS\be_1}{\mu_1}-\frac{N\be_2^H\bS\be_2}{\mu_2}\nn\\
        %&=& B + \frac{N}{\sigma^2} \sum_{k=1}^L \left[ \frac{\mu \lambda_k}{\mu \lambda_k + \sigma^2} \be_1^H\bS_k\be_1 -\sigma^2\ln(\mu \lambda_k + \sigma^2) \right] \nn \\
        &=& B + \frac{N}{\sigma^2} \left[ \be_1^H\bS( \mu ) \be_1 -\sigma^2\sum_{k=1}^L \ln(\mu \lambda_k + \sigma^2) \right] \ ,\nn
        \eea
        where $B$ does not depend on the parameters $\btt$ and
        \beq
        \bS ( \mu ) \ \define \ \sum_{k=1}^L \frac{\mu \lambda_k}{\mu \lambda_k + \sigma^2}  \bS_k %\ \define \ \begin{bmatrix}
        %S_{11}( \mu ) & S_{12}( \mu )\\
        %S_{21}( \mu ) & S_{22}( \mu )
        %\end{bmatrix}
        \eeq
        is the matrix in \er{samp_covMP}. To see this, let
        $\Lambda \define {\rm diag}( \lambda_1 , \ldots, \lambda_L)$ and observe
        \bea\label{3Sijdef2}
        [\bS ( \mu )]_{ij} & \define & \sum_{k=1}^L \frac{\mu \lambda_k}{\mu \lambda_k + \sigma^2} [ \bV \bY_i \bY_j^H \bV^H]_{kk} \nn \\
        & = & \sum_{k=1}^L  \left[ \mu \Lambda \left(\mu \Lambda + \sigma^2 \bI_L \right)^{-1} \bV \bY_i \bY_j^H \bV^H \right]_{kk} \nn \\
        %& = & \sum_{k=1}^L  \left[ \mu \bC_\bH \left(\mu \bC_\bH + \sigma^2 \bI_L \right)^{-1} \bY_i \bY_j^H \right]_{kk} \nn \\
        & = & {\rm Tr}  \left[ \mu \bC_\bH \left(\mu \bC_\bH + \sigma^2 \bI_L \right)^{-1} \bY_i \bY_j^H \right] \ .
        \eea
        
        To find maximum-likelihood estimates of $F$ and $\sigma_h^2$, we proceed in two steps: First, we find conditions on $\mu$ and $\be_1$ that achieve the maximum in \er{3llrMP}. Second, we use \er{eigensystemMP} to translate these conditions into values of
        $F$ and $\sigma_h^2$.
        
        For each $\mu$, the maximum of \er{3llrMP} over $\be_1$ is a unit eigenvector corresponding to the largest eigenvalue of $\bS ( \mu )$. 
        From \cite[eq.~24]{wu_pca}, it is shown this eigenvalue is $\eta ( \mu )$ in \er{3etaMP}.
        %\beq\label{3eta1MP}
        %\eta_1 ( \mu ) \ = \frac{S_{22}(\mu)+S_{11}(\mu) + \sqrt{(S_{11}(\mu)-S_{22}(\mu) %)^2+4|S_{12}(\mu)|^2}}{2} \ .
        %\eeq
        It follows that the maximum-likelihood estimate of $\mu$ is
        \beas
        \hat{\mu} & \define & {\rm arg} \max_{\mu \geq 0} \left[ \eta ( \mu ) - \sigma^2 \sum_{k=1}^L \ln(\mu \lambda_k + \sigma^2) \right]  \ ,
        %& = & {\rm arg} \max_{\mu \geq 0} \left[ \eta_1 ( \mu ) - \sigma^2 \ln {\rm det} [\mu %\bC_\bH + \sigma^2 \bI_L ] \right] \nn \\
        \eeas
        which equals \er{3muhatMP}, since $\sum_{k=1}^L \ln(\mu \lambda_k + \sigma^2) \ = \ \ln {\rm det} [\mu \bC_\bH + \sigma^2 \bI_L ]$.
        
        Finally, we translate these conditions into values of $F$ and $\sigma_h^2$: If $\hat{\mu}=0$, $\bS( \hat{\mu} )$ vanishes and $\ln p( \bY_1 , \bY_2 ; \btt)$ does not depend on $F$. From \er{eigensystemMP}, it follows the likelihood is maximized by $\hat{\sigma}_h^2 = 0$ and any value of $F$; In particular, \er{3MLEMP} maximizes the likelihood. However, if $\hat{\mu} > 0$, then $\bS( \hat{\mu} )$ is not zero and $\be_1$ must be an eigenvector of $\bS( \hat{\mu} )$ corresponding to $\eta ( \hat{\mu} )>0$. For $E_1 \neq 0$, the unique solution to the equations $\hat{\mu} =\sigma_h^2 (1+|F|^2)$ and $\be_1 = (E_1,E_2)^T$ in \er{eigensystemMP} is given by \er{3MLEMP}. For $E_1=0$, no finite $F$ solves these equations; rather, the solution (and maximum) is approached in the limit as $F \rightarrow \infty$. \hfill $\diamond$
    \end{proof}
    
    Theorem \ref{3ThrmMLMP} reduces finding the ML estimators to solving a scalar optimization. In general, this optimization must be done numerically. If the fading channel is i.i.d., i.e., $\bC_\bH=\bI_L$, then \er{3MLEMP} is in closed-form \cite[eq.~22]{wu_pca}.  These i.i.d.  ML estimators are the method of moments (MM) estimators in correlated fading, and serve as a reference to \er{3muhatMP}. 
    
    The entries of the Fisher information matrix (FIM) have been derived, for $1\leq i, j\leq 2$, using \cite[pg. 529]{kay} and extension of (15.60) in Kay \cite[pg. 531]{kay},
    \beq
    [\boldsymbol{\mathcal{I}(\btt)}]_{ij} \ = \ N\cdot \sum_{m=1}^{L}\Trace\left[\bC_k^{-1}\frac{\partial \bC_k}{\partial \theta_i^*}\bC_k^{-1}\frac{\partial \bC_k}{\partial \theta_j}\right] \ ,
    \eeq
    where $\bC_k$ is given in \er{3Ckdef}. We derive the FIM as
    \beq\label{FIM}
    \boldsymbol{\mathcal{I}}(\btt) \ = \ \sum_{k=1}^{L}\frac{N(1+|F|^2)\lambda_k^2}{\left[\lambda_k\sigma_h^2(1+|F|^2) + \sigma^2\right]^2}\boldsymbol{\mathcal{F}}_k \ ,
    \eeq
    where we define
    \beq
    \boldsymbol{\mathcal{F}}_k \ \define \ \begin{bmatrix}
        (\sigma_h^2)^2\left(\frac{\lambda_k\sigma_h^2}{\sigma^2}+1\right)& F\sigma_h^2 \\
        F^*\sigma_h^2& 1+|F|^2
    \end{bmatrix} \ .
    \eeq
    For any unbiased estimators $\hat{\btt}$, the classical Cram\'er-Rao bound (CRB) is then calculated as the inverse of FIM,
    \beq\label{CRB}
    E\left[\left(\hat{\btt}-\btt\right)\left(\hat{\btt}-\btt\right)^H \right]\ \geq \ \boldsymbol{\mathcal{C}}(\btt) \ = \ \boldsymbol{\mathcal{I}}^{-1}(\btt) \ .
    \eeq
    The CRBs for general channel correlation $\bC_\bH$ will be evaluated numerically in the next section.
    
    \section{Numerical Results}\lb{3secV}
    In this section, we explore the performance of estimators in the previous section through numerical examples. Consider a narrow-band MISO communications system with $N=4$ transmit antennas, whose carrier frequency is 900 MHz. 
    The duration of each packet equals a slot of 5G NR (New Radio), i.e., $T_s=1$ ms.
    Block-fading channel is assumed, such that during one packet, the channel remains the same, but it generally varies from packet to packet \cite{bigu}. Other settings follow from \cite[Sec.~IV]{wu_pca}. 
    Based on \er{3observationsMP}, the average post-detection signal-to-noise ratio (SNR) of a received symbol is
    \beq\label{eq_SNR}
    {E\left[\left|\bh^T\bx\right|^2\right]} \ = \ P \sigma_h^2 \ .
    \eeq
    
    Assume Clarke's model for the normalized channel correlation matrix $\bC_\bH$ \cite{zhen,badd}. 
    The maximum Doppler frequency is
    $f_d\define v/\lambda$,
    where $v$ is the velocity  of the fastest moving scatterer and $\lambda$ the wave-length of the carrier frequency \cite[Sec.~IV]{wu_det}.
    
    %%%%%%%%%%%%%%%%%%%%%%%%%%%
    %
    %		Comparison between F_ML and F_MM in correllatd fading, N = 4, L=10
    %
    %%%%%%%%%%%%%%%%%%%%%%%%%%%
    \begin{figure}
        \vspace{-12pt}
        \centering
        \includegraphics[width=.5\textwidth, keepaspectratio=true]{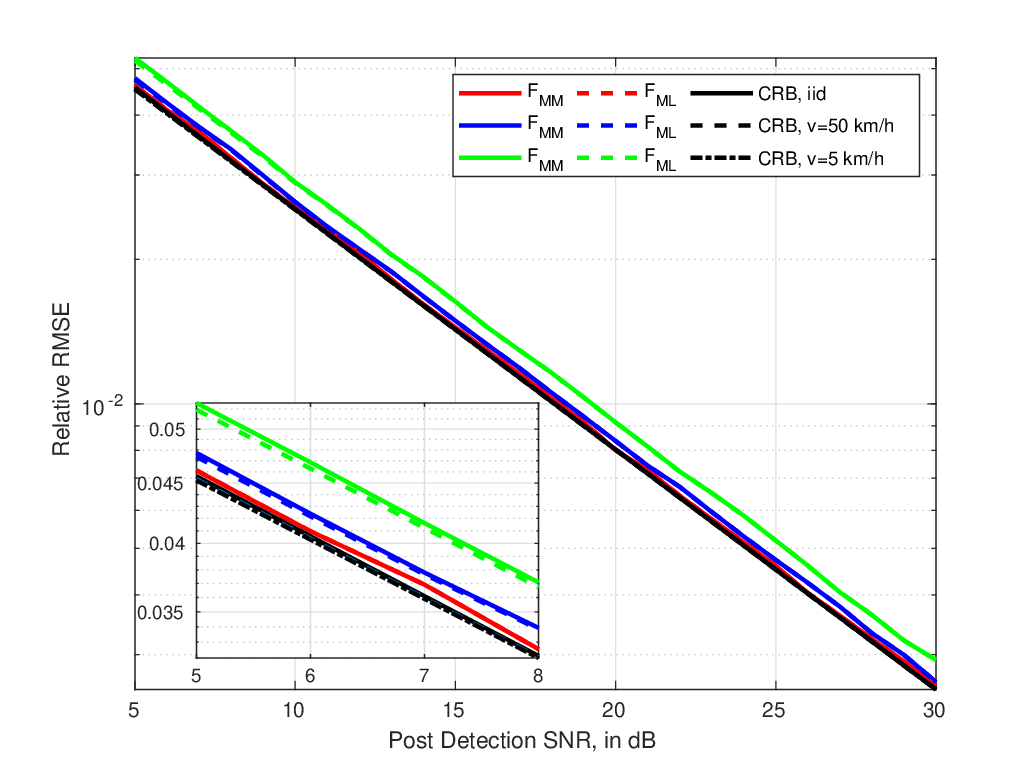}
        \vspace{-18pt}
        \caption{Impedance Estimation under Different Fading.}
        \vspace{-12pt}
        \label{fig_FML_FMM}
    \end{figure}
    
    In Fig.~\ref{fig_FML_FMM}, we compare the behavior of our derived ML estimator $\hat{F}_{ML}$ in \er{3MLEMP} against a reference estimator \cite[eq.~22]{wu_pca}, which we call $\hat{F}_{MM}$. Relative root-mean-square error (RMSE) is used as the metric for performance comparison.  Different fading conditions, fast (i.i.d.), medium ($v=50$ km/h), and slow fading ($v=5$ km/h) are assumed, each with $L=10$  packets. For i.i.d. fading, $\hat{F}_{ML}$ and $\hat{F}_{MM}$ are identical, and their curves completely overlap. In the other fading conditions, the optimal $\hat{F}_{ML}$ exhibits negligible improvement over the simple $\hat{F}_{MM}$ for all SNR considered. 
    Also, the 1 dB gap between $\hat{F}_{ML}$  (or $\hat{F}_{MM}$) and the Cram\'er-Rao bound in \er{CRB} under slow fading is because the CRB could be loose for finite samples. When the fading channel is less correlated, more significant eigenvalues of $\bC_\bH$ exist and hence more independent observations for estimating $F$. This explains the narrower gap to the CRB for fast (i.i.d.) and medium fading conditions. To this end, a rule of thumb is, to be 1 dB within the CRB, a combined 4 orders of diversity, temporal and/or spatial, is needed.  
    Note $\hat{F}_{MM}$ can be obtained in closed-form via direct calculation, but $\hat{F}_{ML}$ is generally found via iterative numerical methods, e.g., a line search. Thus, practical systems may choose $\hat{F}_{MM}$ over $\hat{F}_{ML}$ for a better performance-complexity trade-off. 
    
    \begin{figure}
        \vspace{-12pt}
        \centering
        \includegraphics[width=.5\textwidth, keepaspectratio=true]{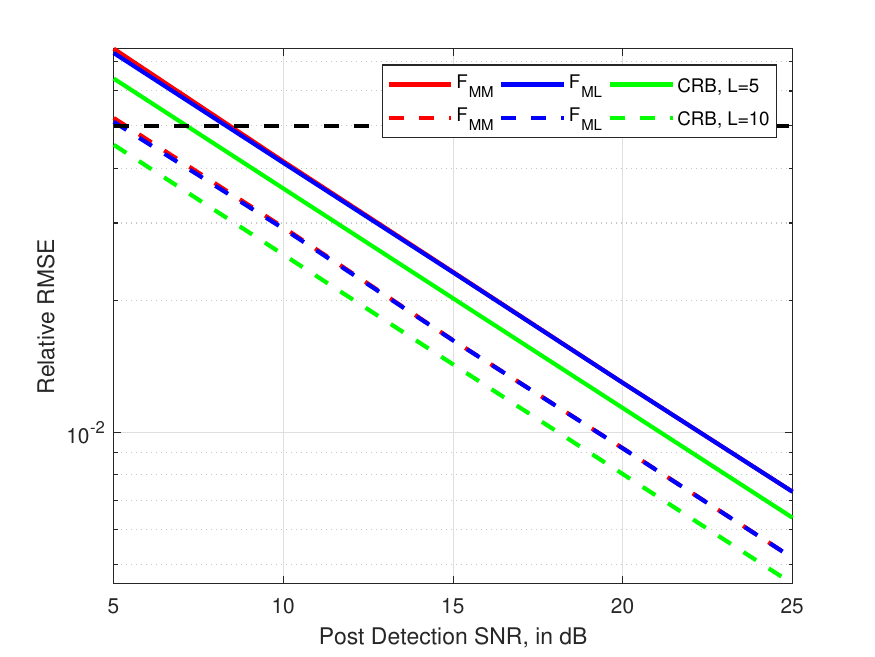}
        \vspace{-18pt}
        \caption{Impedance Estimation under Slow Fading.}
        \vspace{-12pt}
        \label{fig_FML_FMM_slow}
    \end{figure}
    
    In Fig.~\ref{fig_FML_FMM_slow}, we consider a slow fading scenario with packet lengths, $L=5, 10$. Although both $L=5$ to 10 packets lead to only one significant eigenvalue out of the channel correlation matrix $\bC_\bH$,  the doubling in power results in a drop in RMSE by about 3 dB. Moreover, since slow fading is the most different fading condition than i.i.d. fading (where $\hat{F}_{MM}$ is the ML estimator), if the generally optimal $\hat{F}_{ML}$ fails to demonstrate a sizable gain over its counterpart in $\hat{F}_{MM}$, then the closed-form $\hat{F}_{MM}$ may be preferable due to its more efficient implementation in practical systems. 
    
    \section{Conclusion}\lb{3secVI}
    In this paper, we formulated the antenna impedance estimation problem at a MISO receiver in classical estimation. 
    We derived the maximum-likelihood estimator for antenna impedance under generally correlated Rayleigh fading channels. This MLE can be found via scalar optimization. By comparing to a reference, the method of moments (MM) estimator derived in a prequel, we observed the ML and MM estimators exhibit similar RMSE and both approach their CRB's given sufficient degrees of diversity, spatial and/or temporal. A rule of thumb is four degrees of diversity are needed to be within 1 dB of CRB. The MM estimator demonstrated an overall better performance-complexity trade-off. These findings suggest a fast principal-components-based algorithm to estimate antenna impedance in real time for all Rayleigh fading conditions. 
    
    A future direction might be to evaluate the benefit of our derived estimators using system-level metrics, e.g., ergodic capacity. 
    %	References
    \bibliographystyle{unsrt}
    
\end{document}